\PassOptionsToPackage{sort}{natbib} 
\documentclass[prl,twocolumn]{revtex4-1}

\usepackage{amsmath}
\usepackage{graphicx}
\usepackage{braket}
\usepackage{placeins}
\usepackage{xcolor}

\usepackage{amsthm}

\DeclareMathOperator{\Tr}{Tr}

\newtheorem{prop}{Proposition}

\usepackage{amssymb}

\newenvironment{equation-aligned}{\begin{equation}\begin{aligned}}{\end{aligned}\end{equation}}
\begin{document}
\title{Universal witnesses of vanishing energy gap}
\author{Konrad Szyma\'nski$^1$ and Karol \.Zyczkowski$^{1,2}$}

\address{$^1$ Institute of Theoretical Physics,
    Uniwersytet Jagiello\'{n}ski, 30-348 Krak{\'o}w, Poland}
\address{$^2$ Centrum Fizyki Teoretycznej PAN, 02-668 Warszawa, Poland}		

\date{\today}

\begin{abstract}

Energy gap, the difference between the energy of the ground state of a given Hamiltonian and the energy of its first excited state, is a parameter of a critical importance in analysis of phase transitions and adiabatic quantum computation. We present a concrete technique to determine the upper bound for the energy gap of a Hamiltonian $H_0$ based on properties of the set of expectation values of $H_0$ and an additional auxiliary Hamiltonian $V$. This formalism can be applied to obtain an effective
 criterion of gaplessness, which we illustrate with a concrete example of 
the XY model  -- a physical system with vanishing energy gap.
\end{abstract}
\maketitle

\section{Introduction}
The connection between geometry and quantum phase transitions has been studied extensively with various approaches, which include  geometric phase formalism \cite{berry1984quantal}
and studies of properties of the entropy of low-energy states  \cite{masanes2009area}.
 One of the geometry-related methods is based on the notion of {\sl numerical range} 
 -- the set of joint expectation values of several observables among the same quantum state.
 It has been applied in the analysis of uncertainty relations
 \cite{schwonnek2017state,szymanski2019geometric,sehrawat2019uncertainty}, 
 detection of quantum entanglement \cite{czartowski2019separability}, 
 and generalized Wigner functions \cite{schwonnek2020wigner}. 

The link between  the numerical range and quantum phase transition is of special interest, as it connects the static properties of the  set of quantum states (being an image of a linear map) and the dynamical properties of a Hamiltonian family depending on a single parameter, $H_\lambda=H_0+\lambda V$. 
One of the properties with important implications in quantum computing \cite{bachmann2017adiabatic}, state distinguishing, quantum speed limits \cite{marvian2015quantum}, and properties of ground states \cite{hastings2006spectral,bachmann2012automorphic, bravyi2010topological} 
is the value of the energy gap: the energy difference between the ground state and 
the first excited state. There are various ways of estimating the energy gap
 for various kinds of systems: finite-size criteria \cite{knabe1988energy,gosset2016local,lemm2020existence,lemm2019spectral} 
 take into account scaling of the energy gap with dimension of the system
 at finite size  to estimate its asymptotic value. 
 Martingale methods \cite{nachtergaele1996spectral} use 
 other properties of Hamiltonians of a finite size. Using imaginary time propagation \cite{jones2019variational} and methods based on density matrix renormalization group \cite{chepiga2017excitation} the gap can be estimated using numerical methods.
  There exist also techniques  to determine bounds for the  energy gap
 applicable for a certain classes of Hamiltonians
  \cite{levine2017gap,movassagh2017generic,affleck1988valence,lieb1961two,else2020topological}.

In this work, we explore the connection between phase transitions in systems described by a parameterized Hamiltonian $H_\lambda$ and numerical ranges further and provide a new link between the two: the vanishing energy gap is related to the geometrical features of the numerical range of operators $H_0$ and $V$. Hence the geometical properties
of the analyzed sets on the plane is sufficient to determine the gaplessness for certain classes 
of Hamiltonians. We illustrate this connection using a well-studied gapless system: 
a chain of interacting spins forming the XY model \cite{lieb1961two, chen2015discontinuity}.

\section{Preliminaries: joint numerical range}

The methods presented here rely on the properties of low-dimensional projections of the set of density matrices of fixed dimension $d$ (up to affine transformations). Let us recall the definition of the \emph{(joint) numerical range} of $k$ Hermitian operators $A_1, \ldots, A_k$ of order $d$: the numerical range $W(A_1, \ldots, A_k)$ is the set of {simultaneously allowed expectation values} taken over all mixed states \cite{GR12,DGHMPZ11}
\begin{equation}
W(A_1, \ldots, A_k) = \{ (\Tr \rho A_1, \ldots, \Tr \rho A_k) : \rho \in \mathcal{M}_d\},
 \end{equation}

where $\mathcal{M}_d$ denotes the set of density operators of size $d$.

The object $W$ defined above is  a convex subset of $\mathbb{R}^k$, 
the boundary of which is formed by the images of ground states of combinations of $A_1, \ldots, A_k$ --
for an example with $k=2$ see Fig. \ref{ground_state}.
The latter property is useful in the analysis of quantum phase transitions happening at zero temperature \cite{spitkovsky2018signatures,chen2017joint,chen2017physical}.

\section{Boundary of numerical range and quantum phase transitions}
Boundary $\partial W$ of the numerical range $W$
 contains information related to the ground states of Hamiltonians built as combinations of the input operators. This is a direct consequence of the fact that points on the boundary of a convex set are maximizers of linear functionals, formally described in the following statement.
\begin{prop}
\label{prop:groundboundary}
 Let us denote the boundary of the numerical range of a collection of $k$ 
 Hermitian operators $W(A_1, \ldots, A_k)$ by $\partial W$. The point $\vec p=(p_1,\ldots,p_k)$ in $\partial W$ with the inward-pointing normal vector $\vec n=(n_1,\ldots,n_k)$, is an image of the ground state of the Hamiltonian $\sum_{i=1}^k n_i A_i$.
\end{prop}
\begin{proof}
 A point $\vec p$~~having an (inward-pointing) normal vector $\vec n$ on the boundary $\partial W$ of a convex set indicates that it minimizes the functional $\vec n \cdot \vec a$ among points $\vec a \in W$. Since $\vec a$ is a vector of expectation values, this is equivalent to the minimization of $\sum_{i=1}^k n_i \left< A_i \right>$, an expression equal to $\left< \sum_{i=1}^k n_i A_i \right>$. Expectation value is minimized on the ground state of $\sum_{i=1}^k n_i A_i$, which proves the proposition.\end{proof} 
 
\begin{figure}[!h]
\begin{center}
\includegraphics[width=\linewidth]{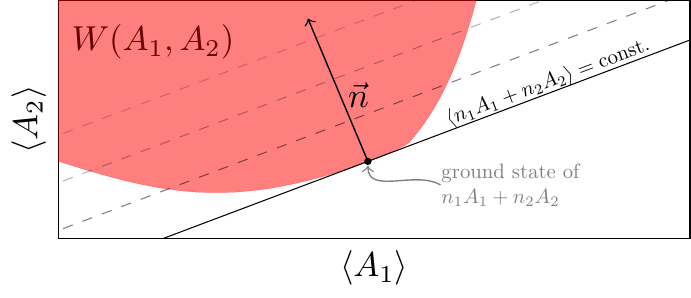}
\end{center}
\caption{Boundary $\partial W$ of a numerical range $W(A_1, A_2)$ is an image of the ground states of linear combinations of $A_1$ and $A_2$.}
\label{ground_state}
\end{figure}

Some properties of the ground states of Hamiltonians composed of $k$ terms, $H=\sum_{i=1}^k n_i H_i$, are visible in the boundary of the numerical range $W(H_1, \ldots, H_k)$. Flat parts of 
the boundary $\partial W$
indicate the degeneracy of the ground state of the Hamiltonian corresponding to the appropriate normal vectors. Cusps, having multiple normal vectors, indicate that the ground state stays constant over a range of parameters \cite{spitkovsky2018signatures}. Let us state here  a simplified version of this result restricted to $k=2$ terms. In such a case, the joint numerical range of two Hermitian observables is
equivalent to the numerical range of a single non-hermitian operator,
  $W(A,B)=W(A+iB)$.
\begin{prop}
\label{prop:cusp}
Consider a numerical range $W(A,B)$ with a cusp, the preimage of which is denoted by $\ket{g}$. Then, the state $\ket{g}$ is a simultaneous eigenvector of $A$ and $B$.
\end{prop}
\begin{proof}
 If a point $(\braket{A}_g,\braket{B}_g) \in \partial W(X,Y)$ forms a cusp of $W$, it corresponds to multiple different support lines with different (normalized) normal vectors. Let us denote two of them by $\vec n$ and $\vec m$. Then by Proposition~\ref{prop:groundboundary}, the state $\ket{g}$ is the ground state of $\vec n \cdot (A,B)$ and of $\vec m \cdot (A,B)$,
\begin{equation-aligned}
n_0 A \ket{g} + n_1 B \ket{g} &= E_{\vec n} \ket{g},\\
m_0 A \ket{g} + m_1 B \ket{g} &= E_{\vec m} \ket{g}.\\
\end{equation-aligned}
\noindent
Since the matrix $\begin{pmatrix} n_0 & n_1 \\ m_0 & m_1 \end{pmatrix}$ is nonsingular if $\vec n \neq \vec m$, this system of equations can be diagonalized to form
\begin{equation-aligned}
A \ket{g} &= \lambda_A \ket{g},\\
B \ket{g} &= \lambda_B \ket{g},\\
\end{equation-aligned}
\noindent
which shows that $\ket{g}$ is an eigenvector of both $A$ and $B$.
\end{proof}
 In the case of a 2-D numerical range $W(A,B)$, there exists a
 relation \cite{caston2001eigenvalues} linking the curvature of the  boundary $\partial W$ 
 with the spectral properties of the nonhermitian operator $A+i B$.

The properties of the numerical range are not only of mathematical interest: convexity implies bounds for the properties of the ground state space of mixed Hamiltonians, which serves as a basis for semidefinite programming (which uses spectrahedra, objects dual to joint numerical ranges \cite{bengtsson2013geometry}), widely used in quantum information theory \cite{schwonnek2017state,szymanski2019geometric}. Sampling of the boundary of the numerical range $W$ has been realized experimentally for three-level quantum systems~\cite{xie2019observing}, confirming the classification established in~\cite{szymanski2018classification}.

\section{Gaplessness witnesses}

Energy gap $\Delta H$
 is the asymptotic difference between the energy $E_0$ of the ground state
 and the  energy $E_1$ of the first excited state, excluding the potential degeneracy of the former.
 For a sequence of Hamiltonians  $H_n$,
  describing systems of increasing size, the system is said to be \emph{gapped} if the value
\begin{equation}
\Delta H=\limsup_{n\rightarrow\infty}\Bigl[  E_1(H_n)-E_0(H_n) \Bigr]
\end{equation}
is nonzero. This expression determines the optimal rate of computation in adiabatic quantum computers \cite{van2001powerful}.
If the gap vanishes, $\Delta(H)=0$, the Hamiltonian $H$ is said to be \emph{gapless}.
This nontrivial
property \cite{cubitt2015undecidability}  
is a signature of a quantum phase transition taking place \cite{zanardi2006ground}. 
We provide here an efficient method to determine whether an arbitrary Hamiltonian $H$ is gapless with the help of a Hermitian perturbation operator $V$, 
which we call a \emph{gaplessness witness} of $H$.

\begin{figure}
\centering
\hspace{6mm}\includegraphics[width=.82\linewidth]{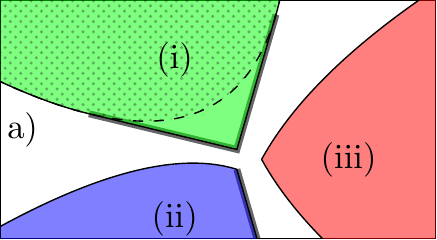}
\includegraphics[width=.9\linewidth]{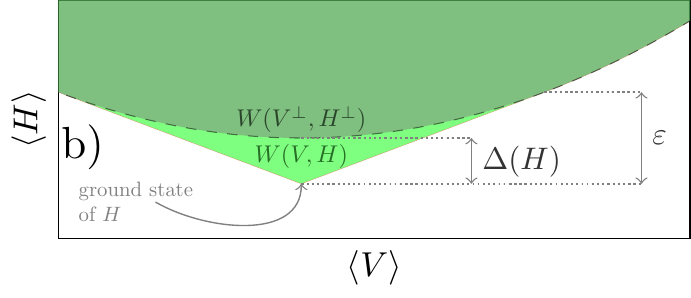}
\caption{a) Different kinds of cusps that could arise in the boundary of the 2D numerical range.  (i) regular cusp with two facets, which can be described as a convex hull of the apex and numerical range restricted to the orthogonal subspace (dotted area); (ii) cusp with one facet, (iii) cusp with no facets. Types (ii) and (iii) can occur only in infinite-dimensional systems. Facets are denoted by thick transparent line. \\
 b) Existence of cusps proves that parts $H, V$ of the total Hamiltonian share a common eigenstate. Geometric features of the segments connected to the cusp determine a bound for the energy gap; if no segments are present -- see case (iii) --
 the Hamiltonian $H$ is gapless.}
\label{fig:kinds}
\end{figure}

\begin{prop}
\label{prop:convhull}
Consider a numerical range $W(H,V)$ with a cusp on its boundary, the preimage of which is the state $\ket{g}$. Let us denote the operators restricted to the one-dimensional subspace spanned by $\ket{g}$ by $H^g$ and $V^g$. Let $H^\perp, V^\perp$ denote  operators restricted to  the orthogonal complement of $\ket{g}$. The numerical range $W(H,V)$ is then a convex hull of the set union of $W(H^g,X^g)$ and $W(H^\perp, V^\perp)$.
\end{prop}

 \begin{proof}
Every state vector can be written as $\ket{\psi}=\sqrt{p} \exp(i\phi)  \ket{g} + \sqrt{1-p}\ket{\psi^\perp}$, where $\ket{\psi^\perp}$ is orthogonal to $\ket{g}$. Since the image of $\ket{g}$ is at the apex of a cusp, the vector is a simultaneous eigenstate of $H$ and $V$ (Proposition~\ref{prop:cusp}). This fact implies that the image of $\ket{\psi}$ in the numerical range lies at
\begin{equation}
\begin{aligned}
(\braket{H}_{\ket\psi},\braket{V}_{\ket\psi})=&p ~(\braket{H}_{\ket{g}},\braket{V}_{\ket{g}})+\\
&\left(1-p\right) (\braket{H}_{\ket{\psi^\perp}},\braket{V}_{\ket{\psi^\perp}}).
\end{aligned}
\end{equation}
The numerical range is thus a convex hull of the point corresponding to the state $\ket{g}$ and the numerical range  $W(H^\perp,V^\perp)$ of operators restricted to the orthogonal subspace:
\begin{equation}
\label{eq:conv}
W(H,V) = \text{conv} \left( W(H^g, V^g), W(H^\perp,V^\perp)\right).
\end{equation}
\end{proof}

The above result, illustrated in Fig. \ref{fig:kinds}, 
implies the following. 

\begin{prop} 
\label{prop:gapless}
If for a Hamiltonian $H$ there exists a Hermitian operator $V$ 
and $t_*>0$ such that 

\begin{enumerate}
\item the ground state $\ket{g(t)}$ of $H+t V$ is constant in~$t~\in~[0, t^*]$,
\item at the point of phase transition ($t=t^*$\hspace{-1mm}), $\braket{H}_{\ket{g(t)}}$ has a jump discontinuity of size
\begin{equation}
\varepsilon = \lim_{t\rightarrow t^*_+} \braket{H}_{\ket{g(t)}}- \lim_{t\rightarrow t^*_-} \braket{H}_{\ket{g(t)}},
\end{equation}  
\end{enumerate}
then the Hamiltonian $H$ has its energy gap $\Delta H$ bounded from above by $\varepsilon$:
\begin{equation}
0\le\Delta(H)\le\varepsilon.
\end{equation}
 If $\left<H\right>_{\ket{g(t)}}$ is continuous at $t=t^*$ for any Hermitian operator $V$ meeting the assumptions (equivalent to $\varepsilon=0$), the Hamiltonian $H$ is gapless.\\
 By contraposition, if a Hamiltonian $H$ is gapped, $\Delta H>0$, then the function $\left<H\right>_{\ket{g(t)}}$ is not continuous at $t=t^*$.
\end{prop}

\begin{proof}
{
The ground state energy of $H^\perp$  is equal to the energy of the
 first excited state of $H$. Thus,  for any state $|\chi\rangle$ 
  orthogonal to $\ket{g}$ 
 the expectation value of $\langle \chi|H|\chi\rangle$ 
 is not smaller than  $E_1(H)$. If $\Delta(H)>0$, this implies the existence 
 of flat segments joining $W(H^g, X^g)$ due to the convex structure of the numerical range (Eq. \eqref{eq:conv}).

This in turn implies the discontinuity of the expectation value  as the Hamiltonian 
parameter $t$ is varied. 
The relevant support points jump from $W(H^g, V^g)$ to $W(H^\perp,V^\perp)$. Since the first excited state minimizes $\langle H\rangle$ on $W(H^\perp,V^\perp)$, the difference of the expectation values of $H$ between the ground state and the new support point must be higher than $\Delta(H)$. 
}
\end{proof}

This method yields promising results in the case of families of Hamiltonians which correspond to increasing system size -- if the estimate $\varepsilon$ decreases to 0 with increasing system size, the system is gapless in the limit of infinite size. In the finite-dimensional case, clearly no spectral gap is possible (set aside the case of degenerated ground state); this is reflected in Theorem 2.1 of \cite{spitkovsky2018signatures} (see also   \cite{bebiano1986nondifferentiabie}), stating that cusps in numerical ranges of finite-dimensional systems are always of type (i) of Fig. \ref{fig:kinds}a.
 
\FloatBarrier
\subsection{Example: XY spin chain model}
\begin{figure}[!h]
\centering
{\includegraphics[width=.9\linewidth]{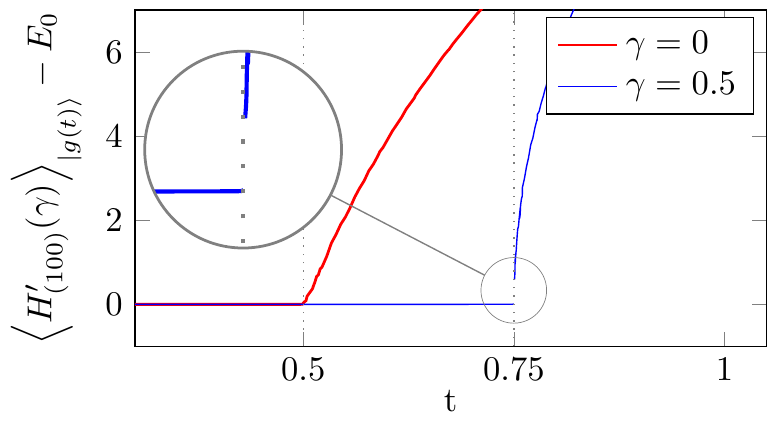}}
{\includegraphics[width=.9\linewidth]{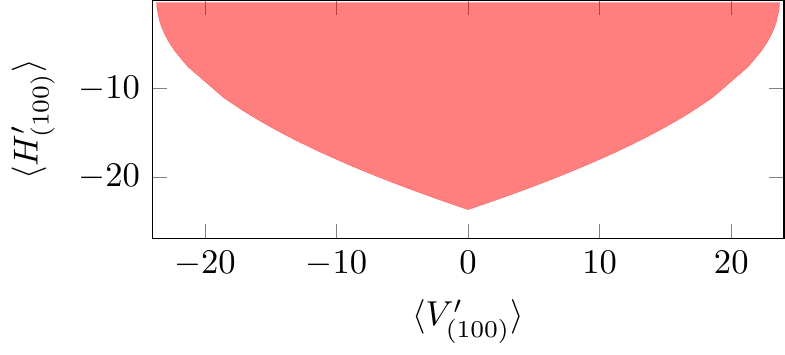}}
\caption{XY model. Top: Expectation value of $H'_{(N)}(\gamma)$ above its ground state energy 
$E_0$ over ground state of $H'_{(N)}(\gamma)~+~t V'_{(N)}$ as a function of $t$  for $H'_{(N)}(\gamma)$ and
 $V'_{(N)}$ defined in Eq. \eqref{eq:taperedhams} (modified versions of original Hamiltonians, see Appendix C). 
  Bounds for gap value for $H'_{(100)}(\gamma)$ follow from  Proposition \ref{prop:gapless}
 and the behavior of expectation value. Continuous (within numerical precision) dependence corresponds to an almost gapless Hamiltonian for $\gamma=0$,  while discontinuity observed for $\gamma=0.5$ allows us to estimate the size of the gap.\\
 Bottom: Joint numerical range
 $W(V'_{(100)},H'_{(100)})$ corresponding to $\gamma=0$.
  The cusp in $\partial W$  at  $\langle V'_{(100)} \rangle=0$, closely resembling kind  (iii) of Fig. \ref{fig:kinds}, suggests almost vanishing energy gap. Calculations for this figures were performed using matrix product states algorithm \cite{joselado2021dmrgpy}.}
\label{fig:xymod}
\end{figure}
Consider the following family of Hamiltonians describing a system of $N$ interacting $1/2$ spins, known as XY model:
\begin{equation}
\label{eq:totalham}
H_{(N)}(\gamma) = -\sum_{i=1}^{N-1} \left( \frac{1+\gamma}{2} \sigma_x^{(i)} \sigma_x^{(i+1)} + \frac{1-\gamma}{2} \sigma_y^{(i)} \sigma_y^{(i+1)}\right).
\end{equation}
In the thermodynamic limit of an infinite spin chain, $N\rightarrow\infty$,
 the Hamiltonian is gapless for $\gamma=0$. The gaplessness is detected by the following Hamiltonian with a three-spin interaction:
\begin{equation}
\label{eq:witnessham}
V_{(N)}=\sum_{i=1}^{N-2} \sigma^{(i-1)}_x \sigma^{(i)}_z \sigma^{(i+1)}_y-\sigma^{(i-1)}_y \sigma^{(i)}_z \sigma^{(i+1)}_x.
\end{equation} 
For $\gamma=0$ the gaplessness of $H$ is detected in the following way: 
to estimate the value of the energy gap of $H_{(N)}$ we apply Proposition \ref{prop:gapless}. It turns out that the ground state of $H_{(N)}$ is an eigenstate of $V_{(N)}$. The numerically determined expectation value $\braket{H}_{\ket{t}}$ as a function of $t$ has a discontinuity, the size of which approaches 0 as $N$ is increased~--~see Fig. \ref{fig:xymod}. This numerical observation is confirmed by the asymptotic analytical treatment presented in Appendix B. Since the size of the gap approaches 0 as $N\rightarrow\infty$, the Hamiltonian $H$ is asymptotically gapless for $\gamma=0$. 

For $\gamma\neq0$, the Hamiltonian $H(\gamma)$ has an energy gap equal to $\gamma$; it is reflected by the growing discontinuity in the expectation value of $H(\gamma)$ with respect to the ground state of $H(\gamma)+t V$.

\section{Concluding remarks}

Estimation of the value of the energy gap is important in various applications, such as
many body physics,  condensed matter physics,  theory of quantum information, 
and adiabatic quantum computing.
In this paper, we provide a method for estimation of the value of the energy gap, which may see practical applications. The technique developed combines geometrical and algebraic properties of the numerical range -- the set of simultaneously attainable expectation values. It allows us to obtain an upper bound for the value of the energy gap with knowledge of the properties of the ground states only. 

Note that the technique proposed here is universal, as it can be applied in a very general setting. While \emph{gaplessness witnesses} -- operators used to detect gaplessness in a geometric way -- exist for every gapless Hamiltonian, generically 'trivial' witnesses can be difficult to implement physically 
as they might be nonlocal. 
Numerical analysis provides evidence that if for a given Hamiltonian $H$, the ground state energy is $E_0=0$, the gaplessness of $H$ is routinely detected by a random observable of the form $H Z H$, where $Z$ is a random Hermitian 
 matrix drawn from the Gaussian Orthogonal Ensemble.

\section{Acknowledgments}
\label{acknowledgment}

It is a pleasure to thank  Marius Lemm, Ramis Movassagh, Piotr Ro{\.z}ek, Stephan Weis, and Jakub Zakrzewski
 for helpful remarks and suggestions. 
Financial support by Narodowe Centrum Nauki under the grant
number DEC-2015/18/A/ST2/00274 and by the Fundacja na rzecz Nauki Polskiej
under the project Team-Net NTQC is gratefully acknowledged.
 
\nocite{spitkovsky2018signatures,xie2019observing,schwonnek2017state,szymanski2019geometric,caston2001eigenvalues,chen2017joint,chen2017physical,sehrawat2019uncertainty,czartowski2019separability}

~
\appendix 
\section{Appendix A: Existence of witnesses and trivial gaplessness}

Provided $H$ is gapless, there always exists at least a single witness. If the ground state energy of $H$ is zero, $H^2 - H$ is capable to show that the gap $\Delta(H)$ vanishes. Furthermore, for $H$ such that $E_0=0$, the 
numerical simulation suggests that a random Hermitian observable $Z$
with probability one yields an effective witness of a vanishing energy gap, $V=HZH$.

The case associated with the gaplessness of $H$ by Proposition \ref{prop:gapless} often appears in various contexts: in one-dimensional Bose-Hubbard model with $N$ sites, the average occupation $\left< n \right> / N$ is constant in the Mott insulator phase and it varies continuously as the hopping rate increases. Similar behavior is observed with the average spin $\left< S_z \right> / N$ in the  XY model interacting with an external field. In these cases, the apparent gaplessness is trivial, as it pertains to operators averaged over the entire system. While the total spin operator $S_z = \sum_{i=1}^N \sigma_z^{i}$ is not gapless, the average spin $S_z / N = \sum_{i=1}^N \sigma_z^{i} /N $ is: it has homogeneously distributed eigenvalues ranging from $-1/2$ to $1/2$ with a spacing of $1/N$. 

\section{Appendix B: Analytical treatment of the XY model witness}
\newcommand{\Sx}[1]{\sigma_{x}^{(#1)}}
\newcommand{\Sy}[1]{\sigma_{y}^{(#1)}}
\newcommand{\Sz}[1]{\sigma_{z}^{(#1)}}
\newcommand{\cc}[1]{c_{#1}}
\newcommand{\cd}[1]{c^\dag_{#1}}
\newcommand{\Sc}[1]{\prod_{m<#1} (1-2\cd{m}\cc{m})}
{We wish to determine the ground state energy and excitation spectrum of the Hamiltonian $G_{(N)}$ acting on a finite chain of $N$ sites, parameterized by $t$ and $\gamma$:
\begin{equation}\begin{aligned}
G_{(N)}(\gamma,t)=&\overbrace{-\sum_{i=1}^{N-1} \left( \frac{1+\gamma}{2} \sigma_x^{(i)} \sigma_x^{(i+1)} + \frac{1-\gamma}{2} \sigma_y^{(i)} \sigma_y^{(i+1)}\right)}^{H_{(N)}(\gamma)}+\\
&t \underbrace{\sum_{i=2}^{N-2} \sigma^{(i-1)}_x \sigma^{(i)}_z \sigma^{(i+1)}_y-\sigma^{(i-1)}_y \sigma^{(i)}_z \sigma^{(i+1)}_x}_{V_{(N)}}.
\label{eq:gofn}
\end{aligned}\end{equation}
The free fermion method \cite{batista2001generalized} works remarkably well in this case. To proceed, let us express the spin operators in the following way:
} 
\begin{align}
\Sx{n}&=(\cc{n}+\cd{n})\Sc{n},\\
\Sy{n}&=i(\cd{n}-\cc{n})\Sc{n},\\
\Sz{n}&=(1-2\cd{n}\cc{n}),
\end{align}
where $\cc{},\cd{}$ obey the fermionic anticommutation relations:
\begin{align}
\{\cc{i},\cd{j}\}&=\delta_{ij},~
\{\cc{i},\cc{j}\}=0,~
\{\cd{i},\cd{j}\}=0.
\end{align}
{This set allows us to write the term representing the XY  part of the Hamiltonian as
\begin{equation}
\label{eq:hpart}
\begin{aligned}
H_{(N)}=&-\sum_{i=1}^{N-1} \left(\frac{1+\gamma}{2} \Sx{i}\Sx{i+1}  +  \frac{1-\gamma}{2} \Sy{i}\Sy{i+1} \right)\\
=&-\sum_{i=1}^{N-1} \left(\cc{i}\cd{i+1}+\cd{i}\cc{i+1}\right) + \gamma  \left(\cc{i}\cc{i+1}+\cd{i}\cd{i+1}\right) + \text{b.t}. 
\end{aligned}
\end{equation} 
The boundary term (b.t.), appearing here due to the broken translational symmetry of a finite chain has the form of a long chain of fermion operators. Its contribution to the behavior of the complete Hamiltonian is, however, relatively constant (of order $O(1)$ in all expectation values, while the rest of $H_{(N)}$ grows as $O(N)$) and independent of $N$ -- since we are interested in the asymptotic behavior as $N\rightarrow\infty$, we omit it in later calculations.
The same substitution also simplifies the form of the witness part of the Hamiltonian:
\begin{equation}\begin{aligned}
V_{(N)}=&\sum_{i=2}^{N-2} \sigma^{(i-1)}_x \sigma^{(i)}_z \sigma^{(i+1)}_y-\sigma^{(i-1)}_y \sigma^{(i)}_z \sigma^{(i+1)}_x\\
= &-2 i \sum_{i=2}^{N-2} \cd{i-1}\cd{i+1} + \cc{i-1}\cc{i+1}.
\end{aligned}\end{equation}}
The total Hamiltonian, $G(\gamma,t)$ has thus the form of nearest- and next-nearest neighbor hopping and creation/annihilation of free fermions. In general, such Hamiltonians are diagonalized by expressing the creation and annihilation operators as a sum of creation and annihilation of plane waves, followed by Bogoliubov transformation \cite{altland2010condensed}:

\begin{equation}
\begin{aligned}
c_n^\dagger &= \frac{1}{\sqrt{N}} \sum_{n} \exp (i n k) d_k^\dagger,\\
c_n &= \frac{1}{\sqrt{N}} \sum_{n} \exp (-i n k) d_k.
\end{aligned}
\end{equation}

The resulting Hamiltonian expressed in $d, d^\dagger$ operators reads
\begin{equation}\begin{aligned}
G_{(N)}(\gamma, t) = &\sum_k \cos (k) d_k^\dagger d_k  \\
&+\gamma(\exp(-ik) d_k^\dagger d_{-k}^\dagger + \exp(ik) d_k d_{-k} )\\
&+ t \sin(2k) d_k^\dagger d_k ,
\end{aligned}\end{equation}
where the momentum index $k$ spans $\pm \frac{\pi}{N}, \pm \frac{3\pi}{N}, \ldots, \pm \frac{(N-1) \pi}{N}$. This form is finally diagonalized by a standard Bogoliubov transformation to
\begin{equation}
G(\gamma, t) = \sum_k d_k^\dagger d_k \left( \underbrace{\sqrt{(\cos k + t \sin 2k)^2 + \gamma^2 \sin^2 k}}_{E_k} -\frac12\right).
\end{equation}

The energies $E_k$ pertain to excitations, while the negative sum of them gives the ground state energy:
\begin{equation}
E_0(\gamma, t) = - \sum_k E_k.
\end{equation}
 This result can be applied in the calculation of the joint numerical range, since the knowledge of the ground state energy as a function of parameters is sufficient to determine the boundary points using the Hellmann-Feynmann theorem: if $E_0(\gamma, t)$ is known, then 
\begin{equation-aligned}
\langle V_{(N)} \rangle_{\gamma, t} &= \frac{\partial E_0}{\partial t} ,\\
\langle H_{(N)} \rangle_{\gamma, t} &= E_0(\gamma, t) - t \frac{\partial E_0}{\partial t}.
\end{equation-aligned}
Calculations done this way are necessarily approximate, since in Eq. \eqref{eq:hpart} the boundary terms were omitted. The error of approximation, however, is of order $O(1)$ (since the boundary term consists of a fixed number of operators with bounded norm), while the expectation values grow as $O(n)$. 

Calculations performed this way agree well with numerical investigations performed using the density matrix renormalization group for large $N$ \cite{joselado2021dmrgpy}. This method can be thus used to explore  the joint numerical range associated with XY model (Fig. \ref{fig:xymod}) and its properties, including the geometric determination of the energy gap value (Proposition \ref{prop:gapless}).
\section{Appendix C: Tapered Hamiltonians}
\begin{figure}[!h]
\centering
{\includegraphics[width=.9\linewidth]{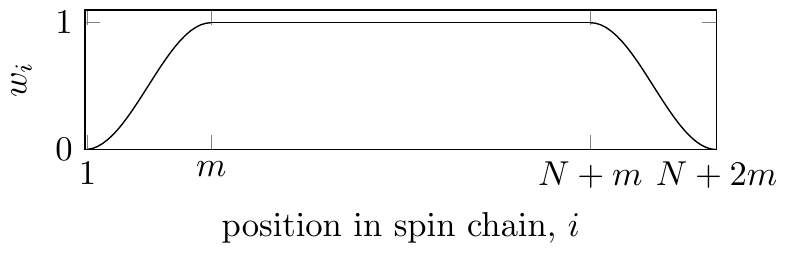}}
\caption{Schematic plot of tapering weight, Eq. \eqref{eq:taper}.}
\label{fig:taper}
\end{figure}
Hamiltonian presented in Eq. \eqref{eq:gofn} needs a slight modification for calculations presented in Fig. \ref{fig:xymod}. The original form is suscepticle to the finite size effects centered at the ends of spin chain. One of the methods to deal with this is to consider cyclic Hamiltonians, in which $\sigma_\mu^{(N+k)}=\sigma_\mu^{(k)}$ -- this brings back translational invariance of eigenstates, but results in suboptimal (yet technically correct and consistent with Proposition 4) plots. We have opted for an alternative: tapering the weights of interactions toward the ends of spin chain. 

The tapered Hamiltonians have very similar form to the originals:
\begin{equation}\begin{aligned}
H'_{(N,m)}(\gamma) &= -\sum_{i=1}^{N+2m-1} \left( \frac{1+\gamma}{2} s_x^{(i)} s_x^{(i+1)} + \frac{1-\gamma}{2} s_y^{(i)} s_y^{(i+1)}\right),\\
V'_{(N,m)}&=\sum_{i=1}^{N+2m-2} S^{(i-1)}_x S^{(i)}_z S^{(i+1)}_y-S^{(i-1)}_y S^{(i)}_z S^{(i+1)}_x.
\label{eq:taperedhams}
\end{aligned}
\end{equation} 

The spin operators $\sigma^{(i)}_\mu$ have been substituted with $s^{(i)}_\mu$ and $S^{(i)}_\mu$. These operators are defined as
\begin{equation}\begin{aligned}
s^{(i)}_\mu &= w_i^{1/2} \sigma^{(i)}_\mu,\\
S^{(i)}_\mu &=w_i^{1/3} \sigma^{(i)}_\mu,
\end{aligned}
\end{equation} 
where $w_i$ is the tapering weight, equal to $1$ in the bulk and approaching $0$ near the spin chain boundary. It appears in different powers in $s^{(i)}_\mu$ and $S^{(i)}_\mu$ to match the spin powers in $H'$ and $V'$ -- resulting  two- and three-site Hamiltonian terms are weighted approximately according to $w_i$.

The weight $w_i$ used in calculations is a sigmoid one, symmetric on both ends (see Fig. \ref{fig:taper}):
\begin{equation}
w_i = \begin{cases} 
      \frac12 \left(1-\cos\frac{\pi (i-1)}{m}\right) & i\leq m \\
      1 & m<i\leq N-m \\
      \frac12 \left(1-\cos\frac{\pi (N+2m-i)}{m}\right) & i>N-m
   \end{cases}.
   \label{eq:taper}
\end{equation}
For plots presented in Fig. \ref{fig:xymod}, tapering distance $m=50$ was chosen. Hence, in addition to $N=100$ spin sites with full weight $w_i=1$ the system was surrounded with $m=50$ sites with decreasing weight on both sides.
\bibliographystyle{own}
 
\bibliography{refs5}~
\end{document}